\documentclass[11pt]{article}
\usepackage{amsmath,amssymb,amsthm,graphicx,subfigure,float,url}

\usepackage{tikz}

\usepackage{parskip} 
\usepackage{pdfsync}
\graphicspath{{fig/}}

\usepackage{mathrsfs} 
\usepackage[colorlinks=true]{hyperref} 
\usepackage{pdfsync}

\usepackage{enumitem}
\usepackage{dsfont} 
\usepackage[T1]{fontenc} 

\def\R{\textrm{I\kern-0.21emR}}
\def\N{\textrm{I\kern-0.21emN}}

\renewcommand{\geq}{\geqslant}
\renewcommand{\leq}{\leqslant}

\renewcommand{\geq}{\geqslant}
\renewcommand{\leq}{\leqslant}

\newcommand {\Chi} {{\bf \raise 2pt \hbox{$\chi$}} }

\newcommand {\sgn} { {\rm sgn} }

\newcommand  \ind[1]  {   \mathds{1}_{#1}   }
\newcommand{\beq}{\begin{equation}}
\newcommand{\eeq}{\end{equation}}
\newcommand{\bea} {\begin{array}{rl}}
\newcommand{\eea} {\end{array}}
\newcommand{\bepa}{\left\{ \begin{array}{l}}
\newcommand{\eepa} {\end{array}\right.}

\newcommand{\bmu}{\begin{multline}}
\newcommand{\emu}{\end{multline}}

\newtheorem{theorem}{Theorem}  
\newtheorem{proposition}{Proposition}
\newtheorem{corollary}{Corollary}

\newtheorem{lemma}{Lemma}

\theoremstyle{definition}\newtheorem{remark}{Remark}
\title{Global stability with selection in integro-differential Lotka-Volterra systems modelling trait-structured populations}

\author{Camille Pouchol\thanks{Sorbonne Universit\'es, UPMC Univ Paris 06, CNRS UMR 7598, Laboratoire Jacques-Louis Lions, F-75005, Paris, France} \thanks{INRIA Team Mamba, INRIA Paris, 2 rue Simone Iff, CS 42112, 75589 Paris, France} \footnotemark[3] \; Emmanuel Tr\'elat\footnotemark[1] \thanks{e-mail: \href{mailto:pouchol@ljll.math.upmc.fr}{pouchol@ljll.math.upmc.fr} (corresponding author), \href{mailto:emmanuel.trelat@upmc.fr}{emmanuel.trelat@upmc.fr}}}

\begin{document}

\newcounter{assum}

\maketitle

\begin{abstract}
We analyse the asymptotic behaviour of integro-differential equations modelling $N$ populations in interaction, all structured by different traits. Interactions are modelled by non-local terms involving linear combinations of the total number of individuals in each population. These models have already been shown to be suitable for the modelling of drug resistance in cancer, and they generalise the usual Lotka-Volterra ordinary differential equations. Our aim is to give conditions under which there is persistence of all species. Through the analysis of a Lyapunov function, our first main result gives a simple and general condition on the matrix of interactions, together with a convergence rate. The second main result establishes another type of condition in the specific case of mutualistic interactions. When either of these conditions is met, we describe which traits are asymptotically selected. 
\end{abstract}

\section{Introduction}
\label{Section1}

\subsection{Biological motivations}

We are interested in the evolution of $N$ populations of individuals, each of which is structured by a continuous \textit{phenotype}, also called \textit{trait}. In each species the phenotype models some continuous biological characteristics (such as the size of the individual, the concentration of a protein inside it, etc). We shall consider both interactions inside a given population and between the populations and we assume that mutations can be neglected. Mathematical modelling and analysis of such ecological scenarios is one purpose of the field of adaptive dynamics, a branch of mathematical biology which aims at describing evolution among a population of individuals, see~\cite{Diekmann2004, Metz2014, Perthame2006} for an introduction to deterministic models. \par

The basis for our model stems from the logistic ODE $\, \frac{d N}{d t} = (r - d N)N$ where $r$ is the net growth rate, $d N$ the logistic death rate due to competition for nutrients and for space by direct or indirect inhibition of proliferation between individuals. Its natural extension to a density $n(t,x)$ of individuals of phenotype $x$ (say in $[0,1]$) is a non-local logistic model 
 \begin{equation}
 \label{Basis}
 \frac{\partial}{\partial t}  n(t,x) = \left(r(x) - d(x) \rho(t)\right) \, n(t,x), 
 \end{equation} with $\rho(t):= \int_0^1 n(t,x) \, dx$ the total number of individuals. \par 
Following~\cite{Zhou2014}, one might consider that these two terms combine both \textit{Darwinism selection} through the intrinsic growth rate $r$, and \textit{Lamarckism induction} through the logistic death rate since it depends on the environment. Note that these models can be derived from stochastic models at the individual level~\cite{Champagnat2008, Diekmann2003, Greene2015}, and more generally measure-valued functions $n$ can be considered~\cite{Gyllenberg2005}. The asymptotic behaviour of the previous model \eqref{Basis} and variants is analysed in~\cite{Greene2014, Lorz2013a, Perthame2006}, and one important property among others is that solution typically tend to concentrate on a few phenotypes, a convergence to Dirac masses in mathematical terms. These models are thus successful at representing the survival of only a few phenotypes, which we will refer to as \textit{selected}.  \par

To account for more complex interactions, one may want to consider a more general non-local logistic term than $d(x) \rho(t)=\int_0^1 d(x) n(t,y)\,dy$, in the form $\int_0^1 K(x,y) n(t,y)\,dy$. The behaviour of such equations strongly depends on how localised the kernel is, and therefore so do the mathematical techniques to analyse them. Indeed, with an added diffusion term, a special case of this situation is the non-local Fisher-KPP equation. When the kernel is localised (small as soon as $|x-y|$ is large), then the solutions typically remain bounded independently of the mutation rate~\cite{Hamel2014}: selection is no longer a feature. This property highlights how differently the solutions behave depending on the kernel, and that some choices are not appropriate for the ecological situation we are concerned with. \par

The non-locality $d(x) \rho(t)$ actually implies that interaction of an individual of phenotype $x$ with other individuals of phenotype $y$ has the same strength $d(x)$ regardless of $y$, because individuals do not only necessarily compete with those which have close phenotypes. As such, our choice can serve as a case study to understand the effect of a blind competition across individuals, essentially mediated by the total density. \par

In~\cite{Greene2014, Lorz2013a}, the biological motivation to use this type of models comes from drug resistance in cancer: the phenotype represents the level of resistance to a given drug and the authors argue that this might be a better approach than a discrete description of the phenotype taking only a finite number of values. Indeed, it can be correlated to some continuous biological characteristics, such as the intracellular concentration of a detoxication molecule, the activity of detoxifying enzymes in metabolizing the administered drug, or drug efflux transporters eliminating the drug~\cite{Chisholm2016}.  \par
This model is further developed in~\cite{Lorz2013a} to incorporate the healthy cell population. Neglecting mutations, it becomes a system of two integro-differential equations of the form 
\begin{equation*}
\begin{split}
     \dfrac{\partial n_1}{\partial t} (t,x)&= \left(r_1(x) - d_1(x)\left(a_{11} \rho_1(t) + a_{22} \rho_2(t)\right) \right) n_1(t,x), \\
     \dfrac{\partial n_2}{\partial t} (t,x)&= \left(r_2(x) - d_2(x) \left(a_{22} \rho_2(t) + a_{21} \rho_1(t)\right) \right) n_2(t,x),
\end{split}
\end{equation*}
with $a_{11}>0$, $a_{22}>0$, $\rho_1$, $\rho_2$ the total number of individuals in the cancer cell and healthy cell populations. The interspecific competition (between the two populations) is taken to be competitive, that is $a_{12}>0$, $a_{21}>0$, but below the intraspecific competition because each cell population is considered to belong to a different ecological niche:
\begin{equation} 
\label{IntraVSInter}
a_{12}<a_{11}, \; a_{21}<a_{22}.
\end{equation} 
In the context of a system arises the central question of persistence (whether asymptotically both species remain), complementing that of identifying which phenotypes are selected.
With the additional difficulty of control terms to represent chemotherapeutic drugs, the asymptotic properties of this model are elucidated in~\cite{Pouchol2016}, and assumption \eqref{IntraVSInter} happens to be crucial. \par
These integro-differential models have therefore already proved their efficiency at helping understanding phenotypic heterogeneity in cancer. The mathematical results available for $N=1$ and $N=2$ for competitive interactions naturally call for generalisations on systems of interacting species with such non-local logical terms based on the total number of individuals. For instance, to study resistance in cancer, one may think of different cancer subpopulations interacting with healthy cells and between them, each one of them being endowed with a specific drug resistance phenotype in a tumour 'bet hedging' strategy~\cite{Brutovsky2013}.
These generalisations, in turn, might help both unravel general principles about the underlying ecological processes, and develop new mathematical techniques to analyse them.

\subsection{The model} 
We consider $N$ populations structured by respective phenotypes $x\in{X_i}$, where $X_i$ is some compact subset of $\mathbb{R}^{p_i}$, with $p_i\in\N^*$, for $i=1,\ldots,N$. Although they model distinct quantities, we abusively denote all variables $x$ to improve readability.

The model writes
\begin{equation}\label{Equations}
\frac{\partial}{\partial t}  n_i(t,x) =
\left(r_i(x) + d_i(x) \sum_{j=1}^N  a_{ij} \rho_j(t)\right) n_i(t,x),\; \; \; \; i=1, \ldots, N,
\end{equation}
where, for $i=1, \ldots, N$, $r_i$ and $d_i > 0$ are functions in $L^\infty(X_i)$, $$\rho_i(t):=\int_{X_i} n_i(t,x) \, dx$$ is the total number of individuals in species $i$, and $a_{ij}\in\mathbb{R}$ are fixed (interaction) coefficients.

The initial conditions are 
\begin{equation}\label{InitialConditions}
n_i(0,\cdot) =n_i^0 \; \; \; \; i=1, \ldots, N
\end{equation}
where each initial density $n_i^0$ is taken to be a non-negative function in $L^1(X_i)$. From now on, we will call these equations \textit{integro-differential Lotka-Volterra equations}.
\par
The matrix $A:=(a_{ij})_{1\leq i,j \leq N}$, called \textit{matrix of interactions}, describes the interactions between the populations: if $a_{ij} >0$, the species $j$ acts positively on the species $i$, and negatively if $a_{ij}<0$. We will also say that the interaction between species $i$ and $j$ for $i \neq j$ is: 
\par
$\bullet$ \textit{mutualistic} if $a_{ij}>0$ and $a_{ji}>0$,
\par
$\bullet$ \textit{competitive} if $a_{ij}<0$ and $a_{ji}<0$,
\par 
$\bullet$ \textit{predator-prey like} if $a_{ij} a_{ji}<0$.

Finally, we will say that the equations are competitive (resp., mutualistic) if $a_{ij}<0$ (resp., $a_{ij}>0$) for all $i \neq j$. 

Another interpretation of the equations is to see them as coupled logistic equations of the form 
\begin{equation}\label{Equations2}
\frac{\partial}{\partial t}  n_i(t,x) =
\big(r_i(x) - d_i(x) I_i(t)\big) \, n_i(t,x),\; \; \; \; i=1, \ldots, N.
\end{equation}
In other words, the species $i$ reacts to its environment through the non-local variable $I_i$ defined for $i=1, \ldots, N$ by 
\begin{equation}\label{LogisticVariable}
I_i := - \sum_{j=1}^N  a_{ij} \rho_j.
\end{equation}
The terms $r_i(x)$ and $d_i(x) I_i$ respectively stand for the net proliferation rate and logistic death rate of individuals in species $i$, of phenotype $x$.

We will also use the notation $R_i(x, \rho_1, \ldots, \rho_N):=r_i(x) + d_i(x) \sum_{j=1}^N  a_{ij} \rho_j$, with which the equations rewrite: 
\begin{equation}\label{EquationsAgain}
\frac{\partial}{\partial t}  n_i(t,x) =
R_i\left(x, \rho_1(t), \ldots, \rho_N(t)\right)n_i(t,x),\; \; \; \; i=1, \ldots, N.
\end{equation}
These models generalise Lotka-Volterra ordinary differential equation (ODE) models~\cite{Baigent2010}: if the functions $r_i$, $d_i$ are all constant (say equal to some $r_i$, and $d_i=1$), then after integration with respect to $x\in{X_i}$, the equations boil down to
\begin{equation}\label{ClassicalLV}
\frac{d}{dt}  \rho_i(t) =
\left(r_i + \sum_{j=1}^N a_{ij} \rho_j(t)\right) \rho_i(t), \; \; \; \; i=1, \ldots, N,
\end{equation}
which we will from now on refer to as \textit{classical $N$-dimensional Lotka-Volterra equations}. Thus another advantage of a logistic term directly defined by $\rho$ is that it makes our model more tractable with respect to the corresponding already well understood ODE models. Conversely, the integro-differential model can be seen as a perturbation of the ODE one where the individuals among a given population are allowed to have different proliferation and death rates depending on their phenotype $x$.

Our goal is to understand the asymptotic behaviour of the solutions to these equations, both in terms of convergence at the level of the total number of individuals $\rho_i$, and in terms of concentration at the level of the densities $n_i$. The first problem is usual in population dynamics while the second is specific to adaptive dynamics and consists of determining which traits asymptotically survive, taking over the whole population. These are then called \textit{Evolutionary Stable Strategies}, and the fact that it is the generic situation has been coined \textit{exclusion principle}. Mathematically, this corresponds to a given density $n_i$ converging to a sum of Dirac masses. For one Dirac mass only, concentration writes, for some $x_0 \in X_i$:
\begin{equation}
n_i(t,\cdot) - \rho_i(t) \delta_{x_0}  \longrightarrow 0
\end{equation}
as $t \rightarrow +\infty$, in the weak sense of measures.  
\par
The more precise aim of this paper is to study the global asymptotic stability (GAS) of what we will call \textit{coexistence steady-states}, namely of possible $\rho^\infty$ with positive components (all species asymptotically survive) such that $\rho$ converges to $\rho^\infty$, because we will see how it determines on which phenotypes the densities concentrate. When it is possible, we will investigate the speed at which convergence and concentration occur. An interesting question within the scope of this paper is also to see if a result of that type is sharp, \textit{i.e.}, to compare the assumptions needed to obtain global asymptotic stability in our generalised setting to those known for classical Lotka-Volterra equations. 

At this stage, we did not make any restrictive assumptions on the matrix $A$. However, it will be clear from the results recalled below in the ODE case and the ones presented in Section \ref{Section2}, that answers to the previous questions are available when interspecific interactions are low compared to the intraspecific ones (reminiscent of \eqref{IntraVSInter}). Thus, we are covering the ecological scenario of each species $i$ having its own niche, but inside which competition (if $a_{ii}<0$) is blind because of the term $a_{ii} \rho_i$.

\paragraph{Notations.} In what follows, $\mathbb{R}^N_{>0}$ will stand for the positive orthant in $\mathbb{R}^N$, the set of vectors whose components are all positive, and we will write $x>y$ when $x-y \in \mathbb{R}^N_{>0}$. We will also use the usual ordering on the set of symmetric matrices: for $A$ a real symmetric matrices, we denote $A \geq 0$ (resp., $A>0$) when $A$ is positive semidefinite (resp., positive definite). Finally, $\mathcal{M}^1(X)$ will denote the set of Radon measures supported in $X$.
 
\subsection{State of the art} 

\paragraph{Classical Lotka-Volterra equations.}
The ODE system \eqref{ClassicalLV} has been extensively studied, dating back to the pioneering works of Lotka and Volterra for two populations of preys and predators~\cite{Lotka1924, Volterra1931}. Since then, many contributions to the analysis of steady-states and their stability have been made, and we refer to~\cite{Murray2002} for an introduction and to~\cite{Baigent2010} for a review. \par
Regarding the global asymptotic stability of coexistence steady-states, a very well-known result due to Goh~\cite{Goh1977} states a simple and very general condition on the matrix $A=(a_{ij})_{1\leq i,j \leq N}$ which ensures convergence to the (unique) coexistence steady-state: 
\smallskip
\begin{theorem}[\cite{Goh1977}] 
\label{GASforLV}
Assume that the equation $A \rho + r = 0$ (where $r\in\mathbb{R}^N$ and $\rho\in\mathbb{R}^N$ are the vectors $(r_i)_{1 \leq i \leq N}$ and $(\rho_i)_{1 \leq i \leq N}$) has a solution $\rho^\infty$ in $\mathbb{R}^N_{>0}$. 
If there exists a diagonal matrix $D>0$ such that $A^T D + D A < 0$, then $\rho^\infty$ is GAS in $\mathbb{R}^N_{>0}$ (and hence is the unique coexistence steady-state) for system \eqref{ClassicalLV}.
\end{theorem}
A result also worth stating is that the mere existence of a unique coexistence steady-state is not enough for it to be GAS. Other steady-states on the boundary of $\mathbb{R}^N_{>0}$ can attract trajectories even in dimension $N=2$. Another possibility is the occurrence of chaotic behaviour even in low dimension as evidenced in~\cite{Vance1978} for $N=3$.
Finally, we mention the more recent work~\cite{Coville2013}, where the authors tackle the problem of GAS for some type of Lotka-Volterra ODEs with mutations. In particular, they obtain GAS of the coexistence steady-state in the case where the logistic variables $I_i$, $i=1, \ldots, N$ all coincide, that is, when they are equal to some variable $I:=\sum_{j=1}^N a_{j} \rho_j(t)$. In such a case, it is proved that the convergence to the equilibrium is exponential. 
The result of GAS is also extended to perturbations of this specific case.

\paragraph{Integro-differential Lotka-Volterra equations.}
The first question for such equations is the existence of a solution for all positive times. This obviously does not hold true in full generality since the ODE $y'=y^2$ is a particular case. Let us first state an existence and uniqueness theorem.
\smallskip
\begin{theorem}
\label{ExistenceUniqueness}
Assume that for a given $n^0 \in \prod_{i=1}^{N} L^1(X_i)$, $n^0\geq 0$, there exists $0<\rho^{sup}$ such that we have an a priori upper bound $\rho(t) \leq \rho^{sup}$ for the functions $\rho_i$ whenever they are defined. Then the Cauchy problem \eqref{Equations}-\eqref{InitialConditions} has a unique solution $n=(n_i)_{1 \leq i \leq N}$, $n \geq 0$, in $C\left([0,+\infty),  \prod_{i=1}^{N} L^1(X_i)\right)$.
\end{theorem}
The proof follows the lines of that given in~\cite[Theorem 2.4]{Perthame2006} for $N=1$ and is detailed in Appendix \ref{AppA}.

In the case of a single equation, the asymptotic behaviour is well understood. For $N=1$, assuming $a_{11}<0$ to avoid blow-up, the equation is simply \[\frac{\partial}{\partial t}  n_1(t,x) = \left(r_1(x) - d_1(x) \rho_1(t)\right) \, n_1(t,x)\] where, without loss of generality, we have set $a_{11}=-1$. The first result is that $\rho_1$ converges.
\smallskip
\begin{theorem}\label{N=1}
Assume some regularity on $X_1$, $r_1$, $d_1$, and $r_1>0$, Then, for any positive continuous initial condition $n_1^0$, $\rho_1$ the function $t\mapsto \rho_1(t)$ is well defined on $[0,+\infty)$ and converges to $\rho_1^M:=\max_{x\in{X_1}}\dfrac{r_1(x)}{d_1(x)}$ as $t\rightarrow+\infty$.
\end{theorem}
This, in turn, completely determines where $n_1$ concentrates.
\smallskip
\begin{corollary} \label{N=1Conc}
Under the previous hypotheses, $n_1(t)$, viewed as a Radon measure on $X_1$, concentrates on the set \[\left\{x\in{X_1}, \; r_1(x) - d_1(x) \rho_1^M=0\right\}\] as $t\rightarrow+\infty$.
If this set is reduced to some $x_1^\infty$, we obtain in particular
\[
n_1(t,\cdot) \longrightarrow  \rho_1^M \delta_{x_1^\infty}
\]
as $t \rightarrow +\infty$ in $\mathcal{M}^1(X_1)$ equipped with its usual weak star topology.
\end{corollary}
A proof of this result can be found in~\cite{Pouchol2016} in the special case $X_1=[0,1]$, and it relies on proving that $\rho_1$ is a bounded variation ($BV$) function on $[0, +\infty)$. Let us stress that when the set on which $n_1$ concentrates is not reduced to a singleton, the steady-state (at the level of $n_1$) is not unique. For example, if the set is made of two points, the repartition of the limiting density on these two points depends on the initial condition, see for example~\cite{Coville2013a}. This is why for this equation and the general equations considered here, there is no hope in proving general GAS results directly at the level of the densities $n_i$.

For a general logistic term $\left(\int_X K(x,y) n(t,y) \, dy\right) n(t,x)$ and a single equation, the asymptotic behaviour is also analysed in detail in both~\cite{Desvillettes2008} and~\cite{Jabin2011}. In the latter, under some suitable assumptions on the kernel $K$, a Lyapunov functional is used to prove that some measure is GAS, in a very specific sense depending on $K$. Similar results can be found in~\cite{Champagnat2010}, where their counterpart for competitive classical Lotka-Volterra equations are also discussed. \par

In the case of integro-differential systems, however, much less is known about the asymptotic behaviour. A Lyapunov functional inspired by \cite{Jabin2011} has been used successfully in~\cite{Pouchol2016} to prove GAS for a competitive system of two populations which writes exactly as \eqref{Equations}. We also mention~\cite{Busse2016} where an integro-differential system of two populations is analysed, and whose form does not fit in our framework. A particular triangular coupling structure allows the authors to perform an asymptotic analysis. \par

The paper is organised as follows. In Section \ref{Section2}, we explain how coexistence steady-states can be identified, allowing us to state rigorously what we mean by GAS for system \eqref{Equations}. We explain why, under the hypothesis of GAS, only some phenotypes are generically selected, and how to compute them. Then, we present the two main results about GAS for such equations. 
Section \ref{Section3} is devoted to the proof of the first result, which applies for any type of interactions and relies on analysing a suitably designed Lyapunov functional. In the specific case of mutualistic interactions, our second main result gives alternative conditions sufficient for GAS. It is presented in Section \ref{Section4}. In Section \ref{Section5}, we conclude with several comments and open questions.

\section{Possible coexistence steady-states and main results}
\label{Section2}

For the rest of the article, we will work with the following assumptions:
\begin{equation}
\label{DataRegularity}
r_i, d_i, n_i^0 \in C(X_i), \; n_i^0>0 \text{ for } i=1, \ldots, N.
\end{equation}
This will simplify statements, but we will be more specific below as to which data our results generalise. 

\subsection{Analysis of coexistence steady-states}

In the context of this system of integro-differential equations, the expression "GAS in $\mathbb{R}^N_{>0}$" must be defined. By that, we mean that there exists $\rho^\infty>0$ such that, whatever the positive continuous initial conditions $n_i^0$ are, $\rho_i$ converges to $\rho_i^\infty$ for all $i$. 

First, let us explain how to compute the possible steady-states at the level of $\rho$, \textit{i.e.}, possible limits $\rho^\infty>0$ for positive continuous initial conditions, with the following topological assumption on the sets $X_i$:
\begin{equation}
\label{Topology}
\forall x \in{\partial X_i}, \; \exists \eta > 0, \; \lambda_{p_i} \left(B(x,\eta) \cap X_i\right) > 0
\end{equation}
where $\lambda_{p_i}$ stands for the Lebesgue measure on $\mathbb{R}^{p_i}$ and $B(x, \eta)$ for the open ball of center $x$ and radius $\eta$. \par
Assume that each $\rho_i$ converges to some $\rho_i^\infty>0$, in which case the exponential behaviour of $n_i(t,x)$ is asymptotically governed by $r_i(x) + d_i(x) \sum_{j=1}^N  a_{ij} \rho_j^\infty$, the sign of which we can analyse as follows. \par 
\begin{itemize}
\item
If this quantity is positive for some $x_0$, let us prove that $n_i(t,x)$ blows up in its neighbourhood, leading to the explosion of $\rho_i$. \par 
If $r_i(x_0) + d_i(x_0) \sum_{j=1}^N  a_{ij} \rho_j^\infty>0$, we choose $\eta > 0$ such that $\lambda_{p_i} \left(B(x_0, \eta)\cap X_i \right) > 0$ and by continuity $r_i(x) + d_i(x) \sum_{j=1}^N  a_{ij} \rho_j^\infty>0$ on $\left(B(x_0, \eta)\cap X_i \right)$. This is possible whether $x_0 \in{int(X_i)}$ or also if $x_0 \in{\partial X_i}$ thanks to (\ref{Topology}). For $\varepsilon >0$ small enough and $t$ large enough (say $t \geq t_0$) such that $r_i(x_0) + d_i(x_0) \sum_{j=1}^N  a_{ij} \rho_j > \varepsilon$, we can write:
\begin{align*}
\rho_i(t) & \geq  \int_{B(x_0, \eta)\cap X_i} n_i(t,x) \, dx  \\
& \geq \int_{B(x_0, \eta)\cap X_i} n_i(t_0,x) \, e^{\int_{t_{ 0}}^t R_i\left(x,\rho_1(s), \ldots, \rho_N(s) \right) \, ds}  \, dx\\
& \geq \lambda_{p_i} \left(B(x_0, \eta)\cap X_i \right) \left( \inf_{B(x_0, \eta)\cap X_i} n_i(t_0,x)\right) \;  e^{ \varepsilon (t-t_0)}, 
\end{align*}
with the right-hand side blowing up as $t\rightarrow+\infty$, which cannot be compatible with the convergence of $\rho_i$.
\item
If $r_i+ d_i \sum_{j=1}^N  a_{ij} \rho_j^\infty$ is negative globally on $X_i$, this clearly implies the extinction of species $i$, which is also incompatible with the convergence of $\rho_i$ to a positive limit.
\end{itemize}
\smallskip
\begin{remark}
It is possible to relax the regularity on both the sets $X_i$ and the data $r_i$ and $d_i$ by working only with points which are both Lebesgue points of $\frac{r_i}{d_i}$ and of Lebesgue density $1$ for $X_i$, see~\cite{Evans2015}. If the functions $\frac{r_i}{d_i}$ are in $L^1(X_i)$, one can indeed check that $r_i+ d_i \sum_{j=1}^N  a_{ij} \rho_j^\infty \leq 0$ $a.e.$ on $X_i$.
\end{remark}
The previous results motivate the following definition:
\[ 
\label{CarryingCapacity}
I_i^\infty := \max_{x\in{X_i}}\dfrac{r_i(x)}{d_i(x)}, \; \; \; \; i=1, \ldots, N.
\] 
With this definition, a steady-state $\rho^\infty>0$ exists if and only if the following assumption holds: 
\begin{equation}
\label{ExistenceSteadyState}
\text{the equation } A \rho + I^\infty = 0 \text{ has a solution } \rho^\infty \text{ in } \mathbb{R}^N_{>0},
\end{equation}
which we assume from now on. \par
 
The previous computations also show that $n_i$ vanishes where $r(x) - d(x) I_i^\infty <0$ which implies the following result:
\smallskip
\begin{proposition}
Assume that assumption \eqref{ExistenceSteadyState} holds, and that $\rho$ converges to the coexistence steady-state $\rho^\infty$. Then, $n_i(t)$, viewed as a Radon measure, concentrates on the set \[K_i:= \big\{x\in{X_i}, \; r_i(x) - d_i(x) I_i^\infty=0\big\}\] as $t\rightarrow+\infty$, for all $i=1, \ldots, N$. \par
If, for some $i$, $K_i$ is reduced to some $x^\infty$, we obtain in particular
\[
n_i(t,\cdot) \longrightarrow \rho_i^\infty \delta_{{x^\infty_i}}
\]
as $t \rightarrow +\infty$ in $\mathcal{M}^1(X_i)$.
\end{proposition}
\smallskip
\begin{remark}
In the hypothesis of global existence and convergence of $\rho$ towards $\rho^\infty$, the previous reasoning actually shows that the concentration is ensured as soon as $n_i^0 \in L^1\left(X_i\right)$ is bounded by below by a positive constant on a neighbourhood of one of the points of $K_i$. For more general hypotheses ensuring concentration, we refer to~\cite{Jabin2011}.

\end{remark}
\smallskip
\begin{remark}
If all the sets $K_i$ are reduces to some singletons $x_i^\infty$, then the dynamics of $\rho$ are asymptotically governed by classical Lotka-Volterra equations concentrated in $\left(x_1^\infty, \ldots, x_N^\infty\right)$, namely 
\begin{equation*}
\frac{d}{dt}  \rho_i(t) =
\left(r_i\left(x_i^\infty\right) + d_i\left(x_i^\infty\right) \sum_{j=1}^N a_{ij} \rho_j(t)\right) \rho_i(t), \; \; \; \; i=1, \ldots, N.
\end{equation*}
For a precise statement, see~\cite{Pouchol2016}.
\end{remark}

\subsection{Analysis of coexistence steady-states}

Our first approach to prove GAS is to further generalise the approach of~\cite{Pouchol2016} in dimension $N$. The main idea is to mix a Lyapunov functional which is inspired by the one designed in~\cite{Jabin2011} and the Lyapunov functional used for classical Lotka-Volterra equations~\cite{Goh1977}, which is the key tool to obtain Theorem \ref{GASforLV}. With some mild regularity assumptions on the data, we obtain the following result: 
\smallskip
\begin{theorem}\label{GAS}
Assume \eqref{ExistenceSteadyState} and that there exists a diagonal matrix $D>0$ such that $DA$ is symmetric and $D A<0$. Then the solution to the Cauchy problem \eqref{Equations}-\eqref{InitialConditions} is globally defined. Furthermore, the solution $\rho^\infty$ to $A \rho + I^\infty=0$ is GAS (and hence, unique).
\end{theorem}
We emphasise that there is no assumption on the type of interactions, \textit{i.e.}, on the sign of the coefficients of $A$. However, a consequence of this result is that $A$ must be such that $a_{ii}a_{jj}>a_{ij}a_{ji}$ for all $i$, $j$. For this result to apply, interactions must therefore be stronger inside species than between them. \par
We also remark that our hypothesis is exactly the one exhibited in~\cite{Champagnat2010} for competitive classical Lotka-Volterra equations. The analysis of the Lyapunov functional allows to determine a speed at which convergence to $\rho^\infty$ and concentration on a given set of phenotypes occur. In dimension $2$, we also analyse more deeply the link between this condition and the one for classical Lotka-Volterra equations, which in most interesting cases happen to be equivalent. 

Our second main result focuses on the special case of mutualistic interactions, and an informal statement of the theorem is the following.
\smallskip
\begin{theorem}\label{GASMut}
Assume \eqref{ExistenceSteadyState}, that for $i=1, \ldots, N$, $r_i>0$ and that for some explicit $0<c_i<C_i$, the matrix $\hat{A}$ defined by $\hat{a}_{ii}:=c_i a_{ii}$ and  $\hat{a}_{ij} = C_i a_{ij}$ for $i\neq j$ is Hurwitz. Then the solution to the Cauchy problem \eqref{Equations}-\eqref{InitialConditions} is globally defined. Furthermore, the solution $\rho^\infty$ to $A \rho + I^\infty=0$ is GAS.
\end{theorem}
Again, this applies to the case of interspecific interactions being higher that intraspecific ones, because a Hurwitz matrix is a matrix such that all its eigenvalues have negative real part and it has to do with diagonally dominant matrices (see Section \ref{Section4}). \par
Because of the hypothesis of mutualism, the system is cooperative, and sub and supersolution techniques can succeed. More precisely, it is possible to prove that all functions $\rho_i$ are BV on $[0, +\infty)$ and this implies their convergence.

\section{General interactions}
\label{Section3}
\subsection{Proof of the main theorem}
In this section, we need slightly more regularity for the data, namely that the functions are Lipschitz continuous: 
\begin{equation}
\label{Regularity}
\text{for } i= 1, \ldots, N,  \; r_i, d_i \in{C^{0,1}(X_i)}.
\end{equation}
We can now restate the first theorem: 
\smallskip
\begin{theorem}\label{GAS_Again}
Assume \eqref{ExistenceSteadyState} and \eqref{Regularity}.
Assume that there exists a diagonal matrix $D>0$ such that $D A$ is symmetric and $DA<0$. Then the solution to the Cauchy problem \eqref{Equations}-\eqref{InitialConditions} is globally defined. \par
Furthermore, the solution $\rho^\infty$ to $A \rho + I^\infty=0$ is GAS with
\begin{equation}
\label{SpeedConvergence}
\rho(t) - \rho^\infty = \mathrm{O}\left(\left(\frac{\ln(t)}{t}\right)^{\frac{1}{2}}\right).
\end{equation}
Concentration of a given $n_i$ occurs at speed $\mathrm{O}\left(\frac{\ln(t)}{t}\right)$, in the following sense: 
\begin{equation}
\label{SpeedConcentration}
\int_{X_i}  m_i(x) R_i\left(x,\rho_1^{\infty}, \ldots, \rho_N^{\infty}\right) n_i(t,x) \, dx  =  \mathrm{O}\left(\frac{\ln(t)}{t}\right). \\
\end{equation}
\noindent
In particular, if $K_i$ is reduced to a singleton $x_i^\infty$, then 
\begin{equation}
\forall \varepsilon>0, \; \; \int_{X_i \setminus B\left(x_i^\infty, \varepsilon\right)} n_i(t,x) \, dx  = \mathrm{O}\left(\frac{\ln(t)}{t}\right).
\end{equation}
\end{theorem}

\begin{proof}

\textit{First step: definition of the Lyapunov functional.} 
\par
For $i=1, \ldots, N$, we choose any measure $n_i^\infty$ in $\mathcal{M}^1(X_i)$ satisfying $n_i^\infty (X_i) = \rho_i^\infty$, which is furthermore concentrated on $K_i$, \textit{i.e.}, 
\begin{equation}
\label{Support}
\text{supp}(n_i^\infty) \subset{K_i}.
\end{equation}
We abusively write integration of functions $g$ against measures $\mu$ as $\int_X g(x) \mu(x) \, dx$. We also set $m_i:= \frac{1}{d_i}$ and define $N$ functions $V_i$ by
\[
V_i(t):=  \int_{X_i} m_i (x) \left[ n_i^\infty(x) \ln\left(\dfrac{1}{n_i(t,x)}\right) + \left( n_i(t,x) - n_i^\infty(x) \right) \right] \, dx.
\]
In what follows, we consider the following Lyapunov functional:
\[
V(t) := \sum_{i=1}^N \lambda_i V_i(t)
\]
where the positive constants $\lambda_i$ are to be chosen later. The diagonal matrix of diagonal entries $\lambda_1, \ldots, \lambda_N$ is denoted by $D$.  \par
\textit{Second step: computation and sign of the derivative.} 
\par
For any $i$, we compute
\begin{align*}
\dfrac{dV_i}{dt} & =  \int_{X_i} m_i(x) \left[- n_i^\infty(x) \dfrac{\partial_t n_i(t,x)}{n_i(t,x)} + \partial_t n_i(t,x) \right] \, dx \\ 
			 & = \int_{X_i} m_i(x) \, R_i\left(x,\rho_1, \ldots, \rho_N\right)  \left[n_i(t,x)- n_i^\infty(x) \right] \, dx\\
			 & = \int_{X_i} m_i(x) \, \left(R_i\left(x,\rho_1, \ldots, \rho_N\right)-R_i\left(x,\rho_1^\infty, \ldots, \rho_N^\infty\right)\right)  \left[n_i(t,x)- n_i^\infty(x) \right] dx\\
			 &  \hspace{12em} + \int_{X_i} m_i(x) \, R_i\left(x,\rho_1^\infty, \ldots, \rho_N^\infty\right)  \left[n_i(t,x)- n_i^\infty(x) \right] \, dx.
\end{align*}
The definition of $m_i$ implies that the first term simplifies as follows
\begin{align*}
\int_{X_i} m_i(x) d_i(x) \, \left[ A\left( \rho-\rho^\infty\right) \right]_i  \left[n_i(t,x)- n_i^\infty(x) \right]  \, dx
  =  \left[ A\left( \rho-\rho^\infty\right) \right]_i  \left( \rho_i - \rho_i^\infty \right).
\end{align*}
For the second term, the choice \eqref{Support} leads to 
\[
B_i(t):= \int_{X_i} m_i(x) \, R_i\left(x,\rho_1^\infty, \ldots, \rho_N^\infty\right)  n_i(t,x) \, dx.\]
The functions $B_i$ are all non-positive by definition of $\rho^\infty$.

Defining the vector $u:= \rho - \rho^\infty$, we arrive at: 
\begin{align*}
\dfrac{dV}{dt} & = 
\sum_{i=1}^N \lambda_i \left[ A\left( \rho-\rho^\infty\right) \right]_i  \left( \rho_i - \rho_i^\infty \right) + \sum_{i=1}^N \lambda_i B_i \\
& = \sum_{i=1}^N  \lambda_i (Au)_i u_i  + \sum_{i=1}^N \lambda_i B_i.
\end{align*}
Thus, we end up with the expression 
\begin{equation}
\label{Derivative}
\dfrac{dV}{dt} = 
u^T (D A) u + \sum_{i=1}^N \lambda_i B_i.
\end{equation}
Since the antisymmetric part of $DA$ does not play any role, this can also be expressed 
\begin{equation*}
\dfrac{dV}{dt} = 
\dfrac{1}{2} u^T (A^T D + DA) u + \sum_{i=1}^N \lambda_i B_i.
\end{equation*}
Thus, we start by assuming that $M:=A^T D + D A<0$ to ensure that $\dfrac{d V}{dt}\leq 0$ and that the convergence of the term $u^T M u$ to $0$ is equivalent to that of $\rho$ to $\rho^\infty$. However, we do not have the usual property $V\geq0$ for Lyapunov functions, so that we cannot yet conclude.

\textit{Third step: estimates on $\frac{dV}{dt}$.} \par
Let 
\[
G:= \frac{1}{2} u^T M u +  2 \sum_{i=1}^N \lambda_i B_i. 
\]
We are going to show that $G$ is non-decreasing.

We denote by $\left\langle u,v\right\rangle$ the canonical scalar product of two vectors $u$ and $v$ in $\mathbb{R}^N$.
\begin{align*}
\dfrac{d}{dt} \left( u^T (D A) u \right) & = \dfrac{d}{dt} \left\langle(DA)u, u \right\rangle \\
 & =  \; \left\langle(DA)\dfrac{du}{dt}, u \right\rangle+ \left\langle(DA)u,\dfrac{du}{dt} \right\rangle.
\end{align*}
For $i=1, \ldots, N$, the derivative of $B_i$ is given by 
\begin{align*}
\dfrac{dB_i}{dt} &= \int_{X_i} m_i(x) R_i\left(x,\rho_1^\infty, \ldots, \rho_N^\infty\right) R_i\left(x,\rho_1, \ldots, \rho_N\right) n_i(t,x) \, dx \\
& =  \int_{X_i} m_i(x) R_i^2\left(x,\rho_1, \ldots, \rho_N\right) n_i(t,x) \, dx \\
&  + \int_{X_i} m_i(x) \left[R_i\left(x,\rho_1^\infty, \ldots, \rho_N^\infty\right) - R_i\left(x,\rho_1, \ldots, \rho_N\right) \right] R_i\left(x,\rho_1, \ldots, \rho_N\right)  n_i(t,x) \, dx  \\
& \geq \big[A \left(\rho^\infty- \rho\right)\big] \int_{X_i} R_i\left(x,\rho_1, \ldots, \rho_N\right)  n_i(t,x) \, dx \\ 
& = - (Au)_i \left(\dfrac{du}{dt}\right)_i
\end{align*}
leading to the bound
\begin{align*}
\dfrac{d}{dt} \left(\sum_{i=1}^N \lambda_i B_i \right)  &\geq -\sum_{i=1}^N \lambda_i (Au)_i \left(\dfrac{du}{dt}\right)_i    \\
&  = - \left\langle (DA) u, \dfrac{du}{dt} \right\rangle.
\end{align*}
Put together, these estimates yield:
\begin{align*}
\dfrac{dG}{dt} &\geq \; \left\langle(DA)\dfrac{du}{dt}, u \right\rangle + \left\langle(DA)u,\dfrac{du}{dt} \right\rangle -2 \left\langle (DA) u, \dfrac{du}{dt} \right\rangle   \\
& = \, \left\langle(DA)\dfrac{du}{dt}, u \right\rangle -  \left\langle(DA)u,\dfrac{du}{dt} \right\rangle.
\end{align*}
The last expression is equal to $0$ if $DA$ is symmetric, in which case $G$ is non-decreasing as claimed.\par
The assumptions that $DA$ is symmetric and that $A^T D + D A<0$ are equivalent to the assumption made for the theorem: $DA$ is supposed to be a symmetric negative definite matrix. 
\par 
As a consequence of the monotonicity of $G$, we get $ u^T (- DA) u  \leq - G(0)$ for all $t$. The left-hand side is the square of some norm on $\mathbb{R}^N$, which means that $\rho$ has to be bounded: these a priori bounds ensure the global definition of the solution to \eqref{Equations}-\eqref{InitialConditions}. 

\textit{Fourth step: a lower estimate for $V$.} 
\par
To estimate $V$ from below, we need a uniform (in $x$) upper bound on the densities $n_i$. Because of the regularity assumption \eqref{Regularity} , there exists $C>0$ such that: 
$$\forall{i=1, \ldots, N}, \; \; \; \forall (x,y) \in{X_i^2}, \; \; R_i\left(y, \rho_1, \ldots, \rho_N\right) \geq R_i\left(x, \rho_1, \ldots, \rho_N\right)- C |x-y|.$$
The constant $C$ can be chosen to be independent of $t$ since the functions $\rho_i$ are bounded. \\
This implies for a given $i$
\begin{align*}
\int_{X_i} n_i(t,y) \, dy &= \int_{X_i} n_i^0(y) \exp\left(\int_0^t R_i\left(y, \rho_1, \ldots, \rho_N\right) \, ds \right) \, dy \\
& \geq \int_{X_i} \dfrac{n_i^0(y)}{n_i^0(x)} \left(n_i^0(x) \exp\left(\int_0^t R_i\left(x, \rho_1, \ldots, \rho_N\right)\right)\right)\exp \left(- Ct |x-y|\right) \, dy \\
& \geq \dfrac{n_i(t,x)}{n_i^0(x)} \int_{X_i} \exp \left(- Ct |x-y|\right) \,dy.
\end{align*}
Computing the integral, we write, thanks to the boundedness of $\rho_i$ $n_i^0$ and  ($C$ has changed and is independent of $t$ and $x$): 
for $t$ large enough, $n_i(t,x) \leq C t $. 
The bound on $V$ follows immediately: 
\begin{equation}
\label{BoundV}
V(t) \geq -C \left (\ln(t) + 1\right).
\end{equation}

\textit{Fifth step: convergence.} 
\par
$G$ bounds $\frac{dV}{dt}$ from above: $\frac{dV}{dt} \leq \frac{1}{2} G$.
Thus $$V(t) - V(0) \leq \frac{1}{2} \int_0^t G(s) \, ds  \leq \frac{1}{2} t G(t)$$ thanks to the third step. 
We can now write $G(t) \geq - C \frac{\ln(t) + 1}{t}$: $G=\mathrm{O}\left(\frac{\ln(t)}{t}\right)$, consequently, each non-positive term it is composed of also vanishes like $\mathrm{O}\left(\frac{\ln(t)}{t}\right)$ as $t\rightarrow +\infty$. \par 
In other words, $ \frac{1}{2} u^T M u$ and each $B_i$ onverge to $0$ as as well $\mathrm{O}\left(\frac{\ln(t)}{t}\right)$. This is nothing but the two first statements \eqref{SpeedConvergence} and \eqref{SpeedConcentration}. \par

For the last statement, we fix $i$ and $\varepsilon>0$. We denote $h_i:=- m_i R_i\left(\cdot, \rho_1^{\infty}, \ldots, \rho_N^\infty\right)$, which is non-negative on $X_i$, and by assumption vanishes at $x_i^\infty$ only. We choose $a>0$ small enough such that 
$a \ind{X_i \setminus B\left(x_i^\infty, \varepsilon\right)} \leq h$ on $X_i$. This enables us to write 
$$ \int_{X_i\setminus B\left(x_i^\infty, \varepsilon\right)} n_i(t,x) \, dx \leq \dfrac{1}{a} \int_{X_i}  m_i(x) R_i\left(x, \rho_1^{\infty}, \ldots, \rho_N^\infty\right) n_i(t,x) \, dx = \mathrm{O}\left(\frac{\ln(t)}{t}\right). $$
\end{proof}

\begin{remark}
The first interesting fact is that the convergence rate of $G$ to $0$, in $\mathrm{O}\left(\frac{\ln(t)}{t}\right)$, is almost optimal in many cases.
Indeed, if the sets $K_i$ are reduced to singletons, there cannot exist any $\alpha>1$ such that this sum vanishes like $\mathrm{O}\left(\frac{1}{t^\alpha}\right)$. This comes from the fact that if it were true, $\frac{d V}{dt}$ would be integrable on the half-line, which would imply the convergence of $V$. This is not possible since each $V_i$ has to go to $-\infty$ as $t$ goes to $+\infty$. \par
This might seem contradictory with the exponential convergence rates obtained in~\cite{Coville2013} for some classical Lotka-Volterra equations, but the Lyapunov functional gives us information on the speed of both phenomena in the sense defined above (through the function $G$) and it does not say whether one of the two is faster.
\end{remark}

\subsection{Sharpness in dimension $2$}
It is clear that if we can find $D>0$ diagonal such that $DA$ is symmetric and $D A<0$, then $A^T D + D A<0$.
The condition that $DA$ should be symmetric is constraining, especially if $N\geq 3$ in which case it imposes some polynomial equalities on the coefficients of the matrix $A$. In dimension $2$, however, the result is sharp in various contexts, as stated in the following proposition.
\smallskip
\begin{proposition}
\label{Sharp}
Assume $N=2$, $a_{11}<0$, $ a_{22} < 0$ and $a_{12}\,a_{21} > 0$. 
Then the following conditions are equivalent. \par
(i) there exists $D>0$ diagonal such that $DA$ is symmetric and $D A<0$; \par 
(ii) there exists $D>0$ diagonal such that $A^T D + D A<0$; \par 
(iii) $\det(A) > 0$.
\end{proposition}

\begin{proof}
(i) implies (ii) as noticed before. \par
Now, let us assume (ii) and compute
$M:=A^T D + D A=$ 
$\begin{pmatrix}
 2  \lambda_1 a_{11} & \lambda_1 a_{12} + \lambda_2 a_{21}  \\
\lambda_1 a_{12} + \lambda_2 a_{21} & 2  \lambda_2 a_{22}  \\
 \end{pmatrix}$, which has positive determinant, \textit{i.e.}, $ \det(M)=4 \lambda_1 \lambda_2 a_{11}a_{22}- \left(\lambda_1 a_{12}+ \lambda_2 a_{21}\right)^2 > 0$. If $\det(A)\leq0$, 
 $\det(M) \leq 4 \lambda_1 \lambda_2 a_{12}a_{21}- \left(\lambda_1 a_{12}+ \lambda_2 a_{21}\right)^2 = - (\lambda_1 a_{12} - \lambda_2 a_{21})^2\leq 0$, a contradiction. \par
Now, if (iii) holds, we take $\lambda_1 := \frac{1}{|a_{12}|} $ and $\lambda_1 := \frac{1}{|a_{21}|} $ for which $D A=$ 
$\begin{pmatrix}
\frac{a_{11}}{|a_{12}|} & \sgn(a_{12})  \\
 \sgn(a_{21})  &  \frac{a_{22}}{|a_{21}|}  \\
 \end{pmatrix}$
is clearly symmetric negative definite, whence (i).
\end{proof}

\section{Cooperative case}
\label{Section4}
\subsection{Some facts about Hurwitz matrices}
We now focus on the cooperative case, \textit{i.e.}, on the case where all off-diagonal elements of $A$ are non-negative. We will also assume that the diagonal elements are negative, since otherwise blow-up clearly occurs: there is intra-spectific competition inside any given species. We will say that such a matrix is \textit{cooperative}. 
\par
In this case, we can hope for stronger results at the level of the integro-differential system because we can use sub and super-solution techniques. For our purpose, the following result on ODEs is sufficient.
\smallskip
\begin{lemma}
For $T>0$ (possibly $T=+\infty$), let $f: [0,T) \times \mathbb{R}^N \longrightarrow \mathbb{R}^N$ be a continous function on $[0,T)\times \mathbb{R}^N$, locally Lipschitz in $x\in{\mathbb{R}^N}$ uniformly in $t\in [0,T)$. Denoting $f(t,x):=\left(f_i(t, x_1, \ldots, x_N)\right)_{1 \leq i \leq N}$, further assume that for all $i=1, \ldots, N$, $f_i$ is non-decreasing with $x_j$ for all $j \neq i$.   \par
Assume that we have a solution $z$ on $[0,T)$ of the following Cauchy problem:
\begin{equation}
\begin{split}
     \dfrac{d z}{dt}& = f(t,z) \\
     z(0) &= z_0, 
\end{split}
\end{equation}
where $z_0\in{\mathbb{R}^N}$, and a function $y$ subsolution to the previous Cauchy problem, \textit{i.e.}, 
\begin{equation}
\begin{split}
     \dfrac{d y}{dt} & \leq f(t,y) \\
     y(0) & \leq z_0.
\end{split}
\end{equation} 
Then $y(t)\leq z(t)$ on $[0,T)$. 
\end{lemma}
When the matrix $A$ is cooperative, it is possible to give an equivalent condition to the one required in Theorem \ref{GASforLV} for GAS in classical Lotka-Volterra equations. Let us explain how, starting with the three following lemmas, the two first of which can be found in~\cite{Baigent2010}.
\smallskip
\begin{lemma}
\label{Useful}
Let $A$ be a cooperative matrix. Then it is Hurwitz if and only it is negatively diagonally dominant, \textit{i.e.}, if and only if there exists a vector $v>0$ such that $a_{ii}v_i + \sum_{j\neq i} a_{ij} v_j < 0$ for $i=1, \ldots, N$.
\end{lemma} 
This first lemma will be useful in its own right in this section. A consequence is that
\smallskip
\begin{lemma}
\label{Positive}
If $A$ is cooperative and $r>0$, $A\rho + r$ has a unique solution in $\mathbb{R}^N_{>0}$ if and only if $A$ is Hurwitz.
\end{lemma}

Finally, it comes from the theory of M-matrices (see~\cite{Plemmons1977} for a review) that 
\smallskip
\begin{lemma}
\label{Equivalence}
Let $A$ is cooperative. Then 
$A$ is Hurwitz if and only if there exists $D > 0$ diagonal such that $A^TD+ D A<0$.
\end{lemma}

An important consequence of these three lemmas is the following rephrasing of Theorem \ref{GASforLV} for classical Lotka-Volterra equations.
\smallskip
\begin{proposition}
\label{ClassicalLVCooperative}
Assume that $A$ is cooperative, $r>0$ and that the equations (\ref{ClassicalLV}) have a unique steady-state in $\mathbb{R}^N_{>0}$. Then the equations are globally defined and this steady-state is GAS.
\end{proposition}
Thus, in the context of cooperation between the species, the requirement that $A$ is Hurwitz is somehow optimal to obtain a GAS coexistence steady-state, since it is already required to have its mere existence, a fact mentioned in~\cite{Goh1979}. We will mainly work with this characterisation (rather than the equivalent one given by Lemma \ref{Equivalence} which we used for a general interaction matrix $A$) because the next results will lead us to modify the matrix $A$: analysing whether it is still Hurwitz or not is easier than checking this equivalent condition.


\subsection{A priori bounds}
For the remaining part of this section, we make the assumption that $r_i$ is positive on $X_i$ for $i=1, \ldots, N$, and we define the lower and upper bounds $0<d_i^m \leq d_i(x) \leq d_i^M$, $0<r_i^m \leq r_i(x) \leq r_i^M$.
\smallskip
\begin{theorem}
\label{GASforLVCoop}
Assume that the matrix $\tilde{A}$ defined by $\tilde{a}_{ii}:= d_i^m a_{ii}$ and $\tilde{a}_{ij} := d_i^M a_{ij}$ is Hurwitz. 
Then the solutions to (\ref{Equations}) are globally defined and bounded.
\end{theorem}

\begin{proof}
First remark that since $\tilde{A}$ is Hurwitz, then so is $A$ from Lemma \ref{Useful}. 

We integrate the equations with respect to $x$ and bound them (through $r_i(x) \leq r_i^M$)
\begin{equation*}
\frac{d}{dt}  \rho_i(t) \leq 
\left(r_i^M + \sum_{j=1}^N a_{ij}\rho_j(t) \int_{X_i} d_i(x) n_i(t,x) \, dx \right) \; \; \; \; i=1, \ldots, N.
\end{equation*}
Since the diagonal elements are negative, the off-diagonal non-negative, we obtain
\begin{equation*}
\frac{d}{dt}  \rho_i(t) \leq 
\left(r_i^M +a_{ii} d_i^m \rho_i + \sum_{j\neq i}^N a_{ij} d_i^M  \rho_j(t)\right) \rho_i(t), \; \; \; \; i=1, \ldots, N.
\end{equation*}
Thus, the vector $(\rho_1, \ldots, \rho_N)$ is a subsolution of the previous system which is nothing but classical Lotka-Volterra equations with interaction matrix $\tilde{A}$. 
Thanks to (\ref{ClassicalLVCooperative}), the solutions to this system are bounded. Thus, so are those of the integro-differential one. 
\end{proof}

\begin{remark}
Note that the assumption that $\tilde{A}$ is Hurwitz reduces to $A$ being Hurwitz in the case of constant coefficients. Thus, this result for boundedness is sharp, since it is exactly the one required for convergence to the coexistence steady-state when the equations at hand are classical Lotka-Volterra equations.
\end{remark}

Using Theorem \ref{GASforLVCoop}, we can thus define $\rho^M>0$ as the GAS steady-state for the system obtained in the previous proof, that is to the equations
\begin{equation*}
\frac{d}{dt}  u_i = 
\left(r_i^M +a_{ii} d_i^m u_i + \sum_{j\neq i}^N a_{ij} d_i^M  u_j(t)\right) u_i(t), \; \; \; \; i=1, \ldots, N.
\end{equation*}
In other words, $\rho^M:= - \left(\tilde{A}\right)^{-1} r^M$ where $r^M$ is the vector $(r_i^M)_{1 \leq i \leq N}$. This means that we can write 
\begin{equation}
\label{UpperBound}
\limsup_{t \rightarrow +\infty} \rho_i \leq \rho_i^M \; \; \; \; i=1, \ldots, N.
\end{equation}
In a similar fashion to the previous proposition, bounding $\rho_i$ away from $0$: 
\begin{equation*}
\frac{d}{dt}  \rho_i \geq 
\left(r_i^m +a_{ii} d_i^M \rho_i + \sum_{j\neq i}^N a_{ij} d_i^m  \rho_j(t)\right) \rho_i(t), \; \; \; \; i=1, \ldots, N,
\end{equation*}
leading to 
\begin{equation}
\label{LowerBound}
\liminf_{t \rightarrow +\infty} \rho_i \geq \rho_i^m \; \; \; \; i=1, \ldots, N.
\end{equation}
where $\rho^m:= - \left(B \right)^{-1} r^m>0$ with $B$, a Hurwitz matrix defined by $b_{ii}:= d_i^M a_{ii}$, $b_{ij}:=d_i^m a_{ij}$ for $i \neq j$ and $r^m:=(r_i^m)_{1 \leq i \leq N}$. 

\subsection{GAS in the mutualistic case}
We can now state the main result: 
\smallskip
\begin{theorem}
\label{GAS_Mut_Again}
Assume $r_i>0$ for all $i=1, \ldots, N$, and that the matrix $\hat{A}$ defined by $\hat{a}_{ii}:=d_i^m \rho_i^m a_{ii}$ and  $\hat{a}_{ij} := d_i^M \rho_i^M a_{ij}$ for $i\neq j$, is Hurwitz. \par 
Then $\tilde{A}$, $A$ and $B$ are also Hurwitz. Furthermore, $\rho^\infty:= - A^{-1} I^\infty$ lies in $\mathbb{R}^N_{>0}$ and it is GAS. 
\end{theorem}
\begin{proof}
The fact that $\tilde{A}$, $A$ and $B$ are also Hurwitz is a direct consequence of Lemma \ref{Useful}. \par 
$\tilde{A}$ being Hurwitz, the solutions are globally defined with $\rho$ bounded thanks to Theorem \ref{GASforLVCoop}.

$A$ being Hurwitz, it is invertible and $\rho^\infty:= - A^{-1} I^\infty$ is in $\mathbb{R}^N_{>0}$ thanks to Lemma \ref{Positive}. \par

Let us now prove that it is GAS. The idea is to prove that each $\rho_i$ is $BV$ on $[0,+\infty)$. Identifying the limit is straightforward, thanks to the reasoning made in Section \ref{Section2}. 

For any $i$,  we define $q_i := \dfrac{d\rho_i}{dt} $ and write $R_i=R_i\left(x, \rho_1, \ldots, \rho_N\right)$ for readability. 
Since $q_i= \dfrac{d\rho_i}{dt} = \int_{X_i} R_i n_i$, we obtain
\begin{align*}
\dfrac{dq_i}{dt} & =  \int_{X_i} R^2_i n_i + \int_{X_i} \left(\sum_{j=1}^N \dfrac{\partial R_j}{\partial \rho_j} q_j   \right)n_i \\ 
			 & \geq \sum_{j=1}^N a_{ij}\left( \int_{X_i} d_i(x) n_i(t,x) \,dx \right)q_j.
\end{align*}
Let $b_i(t):= \int_{X_i} d_i(x) n_i(t,x) \,dx$. The idea is that $\rho_i$ is "mostly" increasing, so we are interested in $(q_i)_-$.
 for which we have the (a.e.) bound
 \begin{equation*}
\dfrac{d(q_i)_-}{dt}  \leq b_i \sum_{j=1}^N a_{ij} q_j \left( -\ind{q_i>0} \right) 
\end{equation*}
On the one hand, 
\[ b_i a_{ii} q_i \left( -\ind{q_i>0} \right) = b_i a_{ii} (q_i)_-.\]
On the other hand, for $i \neq j$, \[ b_i a_{ij} q_j \left( -\ind{q_i>0} \right)\leq b_i a_{ij} (q_j)_-.\]  Combining these two, we get
 \begin{equation*}
\dfrac{d(q_i)_-}{dt}  \leq b_i \left(A (q)_- \right)_i.
\end{equation*}
We fix $\varepsilon>0$ small enough and $t$ large enough (say $t\geq t_0$) such that $\hat{A} + \varepsilon J$ is Hurwitz (where $J$ is the matrix composed of ones only) and such that, for each $(i,j)$, $b_i(t) a_{ij} \leq \hat{a}_{ij} + \varepsilon$. The first requirement is easily derived from Lemma \ref{Useful} since $\hat{A} + \varepsilon J$ is clearly cooperative and negatively diagonally dominant for $\varepsilon>0$ small enough. The second requirement comes from the lower and upper bounds for the functions $\rho_i$ as stated in (\ref{UpperBound}) and (\ref{LowerBound}).

The resulting inequality is 
 \begin{equation*}
\dfrac{d(q_i)_-}{dt}  \leq \left(\left(\hat{A} + \varepsilon  J\right)(q)_- \right)_i,
\end{equation*}

so that $\left((q_1)_-, \ldots, (q_N)_-\right)$ is a subsolution of the system with same initial conditions at $t_0$, given by
 \begin{equation*}
\dfrac{d y} {dt}  = \left(\hat{A} + \varepsilon J\right) y.
\end{equation*}
The solutions to this system go exponentially to $0$ since $\hat{A} + \varepsilon J $ is Hurwitz. 

For any $i$, we have thus proved that $(q_i)_-$ goes to $0$ exponentially. 
Together with the fact that $\rho_i$ is bounded from above, we conclude that it is $BV$ on $[0,+\infty)$. Indeed it holds true that a function $u$ which is both bounded from above and such that $u_- \in L^1([0,+\infty))$ is $BV$ on $[0,+\infty)$, see~\cite[Lemma 6.7]{Perthame2006}. 
Therefore, $\rho$ converges (to $\rho^\infty$). 
\end{proof}

%
%
%

\section{Conclusion}
\label{Section5}

We have analysed the global asymptotic stability properties for integro-differential systems of $N$ species structured by traits $x$ belonging to different trait spaces $X_i$. The coupling comes from a non-local logistic term, which is a linear combination of the total number of individuals $\rho_i$ in each species.  Theses systems generalise the usual Lotka-Volterra ODEs for which many stability analyses are available in the literature. Our main focus has been on the asymptotic behaviour of the functions $\rho_i(t)$ as $t\rightarrow +\infty$, especially towards equilibrium $\rho^{\infty}$ with positive components, \textit{i.e.}, of persistence of all species. In Section \ref{Section2}, we explained how identifying it essentially determines the asymptotic behaviour of the underlying density $n_i$, namely the phenotypes on which the measures $n_i(t, \cdot)$ concentrate in large time.
 
In Section \ref{Section3}, an adequate Lyapunov functional allowed us to state a general result relying solely on an assumption on the matrix $A$, regardless of the type of interactions. For $N=2$, this is essentially a sharp result, but becomes more restrictive for $N \geq 3$. This tool also provided us with convergence rates to equilibrium. In Section \ref{Section4}, we presented another strategy based on a $BV$ bound which yielded a second result of global asymptotic stability, this time for mutualistic equations.

The result of Theorem \ref{GAS_Mut_Again} is partly less general than the one given in Theorem \ref{GAS_Again} because it requires a sign on the coefficients of $A$. However, the set of matrices which satisfy the hypothesis given in the last theorem is an open subset of the set of real matrices $\mathbb{R}^{N \times N}$ in any dimension. This in sharp contrast with the hypothesis of Theorem \ref{GAS_Again}, which, as already mentioned, imposes some polynomial equalities on the coefficients of $A$ as soon as $N \geq 3$. In other words, for a small perturbation of a cooperative matrix for which GAS holds, GAS still holds. In particular, if one has weakly (but mutualistically) coupled equations, GAS holds, whereas Theorem \ref{GAS_Again} does not cover the case of any weakly coupled equations for general interactions, except for $N = 2$.

In both cases, the assumptions fall within the class of matrices which cannot have off-diagonal coefficients which are too high compared to the diagonal ones. The present results thus apply to cases where interactions among individuals of a same species are not only blind because of the term $a_{11} \rho_1$, but also stronger than the interactions between species. In other words, each one of them has its own ecological niche inside which interactions are independent of how given phenotypes $x$ and $y$ are away from another.

Let us remark that the $BV$ method would apply to more general functions $R_i(x,\rho_1, \ldots, \rho_N)$, as long as they are increasing in the variables $\rho_j$, $j \neq i$. However, the Lyapunov functional used in Theorem \ref{GAS_Again} seems to be dependent on the linear coupling chosen here and it is an open problem to generalise our results for other settings. Another open question is about finding whether there are matrices $A$ for which the underlying classical Lotka-Volterra equations converge to the coexistence steady state (for example such that there exists $D>0$ with $A^T D +D A<0$), but for which there is no GAS for the integro-differential system. Numerically at least, we could not build any such case.

We finally mention a natural extension of these integro-differential systems, which has drawn much attention in adaptive dynamics: that is when one adds a mutation term $\mathcal{M}[n]$, and a typical single equation then writes 
\begin{equation}
\label{Mutation}
\frac{\partial}{\partial t}  n(t,x) = \left(r(x) - \int_X K(x,y) n(t,y) \, dy \right) n(t,x)+\mathcal{M}[n](t,x).
\end{equation}

This is usually done either through a second order elliptic operator like \[\mathcal{M}[n](t,x)= \beta \Delta n(t,x),\] with Neumann boundary conditions, or using an integral operator allowing for long-range mutations like \[\mathcal{M}[n](t,x) = \int_X m(x,y) \left(n(t,y)-n(t,x)\right) \,dy.\] 
Finally, let us mention that in some more recent works an advection term is also considered~\cite{Chisholm2016a, Chisholm2015}. The idea is to model stress-induced adaptation: individuals actively adapt to their environment and this can be thought of as an appropriate modelling of Lamarckism induction. \par
Most of the studies on model \eqref{Mutation} have aimed at understanding how small mutations affect the dynamics of the surviving phenotype (when one expects a single Dirac mass) with a vanishing mutation term $\mathcal{M}_\varepsilon$, after a proper rescaling of time~\cite{Barles2009, Lorz2011, Mirrahimi2015}. Without smallness assumptions on the mutations, non-trivial steady-states have been investigated in detail in~\cite{Bonnefon2015, Coville2013a, Leman2015}, and results of GAS have essentially been obtained for $K(x,y) = k(y)$. It would be interesting to investigate the asymptotic stability properties of steady states for more general kernels. To the best of our knowledge, general asymptotic results are indeed still elusive, even in the case explored in our paper, namely when $K(x,y)=d(x)$ for $d$ non constant. It is not clear whether the techniques developed in the present paper can be extended to that setting nor to systems of this form.

\appendix 
\section{Proof of Theorem \ref{ExistenceUniqueness}} 
\label{AppA}
\begin{proof}

The proof is based on the Banach-Picard fixed point theorem. 
We set $T>0$, and define the Banach spaces $Z :=  \prod_{i=1}^{N} L^1(X_i)$ endowed with the max norm, and
$E:=C\left([0,T],Z\right)$ endowed with the norm $\|m\|_E:=\sup_{0\leq{t}\leq{T}}\|m(t)\|_{Z}$.
Finally, we consider the following closed subset: $F:=\left\{m\in{E} \, / \, m\geq 0 \,  \,  \text{and} \,  \,  \|m\|_E\leq{M}\right\}$ where $M > \rho^{\sup}$. 

We now build the application. 
Let $m$ be a fixed element in $F$, and let us define for $i=1, \ldots, N$
\begin{equation*}
  \tilde{\rho_i}(t)=\int_{X_i} m_i(t,x)\,dx.
\end{equation*}
For each $i=1, \ldots, N$ and each fixed $x \in{X_i}$, we consider the solution $\gamma_{i,x}$ to the following differential equation:
\begin{equation}
\label{ODE}
  \left\{
    \begin{array}{ll}
 \frac{d\gamma_{i,x}}{dt}=R_i\left(x, \tilde{\rho}_1(t), \ldots, \tilde{\rho}_N(t)\right) \, \gamma_{i,x} \\
\gamma_{i,x}(0)=n_i^0(x)\\
    \end{array}
\right. 
\end{equation}
which is global on $[0,T]$. 

We then define for all $(t,x)$ in $[0,T] \times X$ and $i=1, \ldots, N$ the function $n_i(t,x):=\gamma_{i,x}(t)$, thus building an application $\Phi$ through $\Phi(m):=n$.

We now show that $\Phi$ maps $F$ onto itself. \\
The equation (\ref{ODE}) can be solved explicitly by 
\begin{equation*}
n_i(t,x)=n_i^0(x)e^{\int_0^t  R_i\left(x, \tilde{\rho}_1(s), \ldots, \tilde{\rho}_N(s)\right)  \, ds},
\end{equation*}
which shows both $n\geq{0}$ and $n\in{E}$. 

Let us fix some $i=1, \ldots, N$ and bound as follows 
\begin{equation*}
\frac{\partial}{\partial t}  n_i(t,x) \leq  \left(\|r_i\|_{L^\infty} + \|d_i\|_{L^\infty} \|A\|_\infty \rho^{\sup}  \right) n_i(t,x).
\end{equation*}
Integrating in $x$, we uncover $\frac{d}{d t} \|n(t)\|_Z\leq C  \|n(t)\|_Z$ for some constant $C>0$, which leads to
\begin{equation*}
 \|n(t)\|_Z\leq \rho^{\sup} e^{CT}.
\end{equation*}
To obtain $n \in{F}$, it only remains to choose $T$ small enough so that $ \rho^{\sup} e^{CT} \leq K$.

The last step is to prove the strong contraction property for $\Phi$. In the following, $C$ will denote various positive constants, which might change from line to line. Let $(m^1,m^2)\in{F^2}$ and $(n^1,n^2)$ its image by $\Phi$. We define $\tilde{\rho}^k$ as before for $k=1,2$. For all $i$, we write
\begin{equation*}
(n_i^1-n_i^2) (t,x) = n_i^0(x) \left[e^{\int_0^t R_i\left(x, \tilde{\rho}^1_1(s), \ldots, \tilde{\rho}^1_N(s)\right)  \, ds} - e^{\int_0^t R_i\left(x, \tilde{\rho}^2_1(s), \ldots, \tilde{\rho}^2_N(s)\right)  \, ds} \right]. 
\end{equation*}
Now, since the argument in the exponentials can be bounded by $C T$, the mean value theorem yields
\begin{align*}
|(n_i^1-n_i^2)| (t,x) & \leq n_i^0(x) e^{C T} \left| \int_0^t  \left[ R_i\left(x, \tilde{\rho}^1_1(s), \ldots, \tilde{\rho}^1_N(s)\right) -R_i\left(x, \tilde{\rho}^2_1(s), \ldots, \tilde{\rho}^2_N(s)\right) \right]  \, ds \right| \\
				  & \leq   \|d_i\|_{L^\infty} \|A\|_\infty n_i^0(x) e^{C T} \left[ \int_0^T   \|\tilde{\rho}^1(s)-\tilde{\rho}^2(s)  \|_\infty \, ds  \right]  \\
				  & \leq  C n_i^0(x) T e^{CT} \|m_1-m_2\|_E.
\end{align*}
This implies after integrating in $x$ and taking the supremum both in $t\in[0,T]$ and in $i=1, \ldots, N$:
\begin{equation*}
\|n^1 - n^2\|_E \leq C \rho^{sup}T e^{C T} \|m^1-m^2\|_{E} .
\end{equation*}
It provides us with the contracting property for $\Phi$ whenever $T$ is small enough.

We conclude by noticing that $T$ has been chosen small independently of the initial data, so that the argument can be iterated on $[0,T]$, $[T,2T]$, etc. 
\end{proof}

{
\bibliography{BibliographyLV}
\bibliographystyle{acm}}

\end{document}